\documentclass[12pt]{article}
\usepackage{amsfonts,amssymb,amsbsy,amsmath,amsthm,latexsym,natbib,graphicx}
\usepackage{epsf}
\usepackage[english]{babel}
\usepackage[utf8]{inputenc}
\usepackage{algorithm}
\usepackage[noend]{algpseudocode}
\usepackage{subfigure,color,comment}
\usepackage{hyperref}

\newtheorem{theorem}{Theorem}
\newtheorem{proposition}{Proposition}
\newtheorem{lemma}{Lemma}

\theoremstyle{definition}

\newcommand{\sfh}{\mathsf{h}}

\newcommand{\thetahat}{\widehat{\theta}}
\newcommand{\muhat}{\widehat{\mu}}

\newcommand{\Rbar}{\overline{R}}

\newcommand{\sfI}{\mathsf{I}}
\newcommand{\sfR}{\mathsf{R}}
\newcommand{\sfS}{\mathsf{S}}

\newcommand{\Prob}[1]{\mathbb{P} \left({#1}\right)}

\newcommand{\Expect}[1]{\mathbb{E} \left[{#1}\right]}

\newcommand{\matM}{\mathsf{M}}
\newcommand{\matQ}{\mathsf{Q}}
\newcommand{\matP}{\mathsf{P}}
\newcommand{\matI}{\mathsf{I}}
\newcommand{\matA}{\mathsf{A}}

\newcommand{\md}{\mbox{d}}

\newcommand{\Coffin}{\mathsf{C}}

\newtheorem{example}{Example}

\title{Direct statistical inference for finite Markov jump processes via the matrix exponential}
\author{Chris Sherlock$^1$\footnote{c.sherlock@lancaster.ac.uk}}
\date{{\small $^1$Department of Mathematics and Statistics, Lancaster University, UK}}

\begin{document}
\maketitle
%\begin{center}
%\textbf{Abstract}
%\end{center}
\begin{abstract}
  Given noisy, partial observations of a time-homogeneous, finite-statespace Markov chain, conceptually simple, direct statistical inference is available, in theory, via its rate matrix, or infinitesimal generator, $\mathsf{Q}$, since $\exp (\mathsf{Q}t)$ is the transition matrix over time $t$. However, perhaps because of inadequate tools for matrix exponentiation in programming languages commonly used amongst statisticians or a belief that the necessary calculations are prohibitively expensive, statistical inference for continuous-time Markov chains with a large but finite state space is typically conducted via particle MCMC or other relatively complex inference schemes.
  
When, as in many applications $\mathsf{Q}$ arises from a reaction network, it is usually sparse. We describe variations on known algorithms which allow fast, robust and accurate evaluation of the product of a non-negative vector with the exponential of a large, sparse rate matrix. Our implementation uses relatively recently developed, efficient, linear algebra tools that take advantage of such sparsity.  We demonstrate the straightforward statistical application of the key algorithm on a model for the mixing of two alleles in a population and on the Susceptible-Infectious-Removed epidemic model.
\end{abstract}

\section{Introduction}
A \emph{reaction network} is a stochastic model for the joint evolution of one or more populations of \emph{species}. These species may be chemical or biological species \cite[e.g.][]{Wilkinson2012}, animal species \cite[e.g.][]{drovandi16}, interacting groups of individuals at various stages of a disease \cite[e.g.][]{AndBritt2000}, or counts of sub-populations of alleles \cite[e.g.][]{Moran1958}, for
example.
The state of the system is encapsulated by the number of each species that is
present, and the system evolves via a set of \emph{reactions}: Poisson processes whose rates depend on the current
state.

 Typically, partial and/or noisy observations of the state are available at a set of time points, and statistical interest lies in inference on the unknown rate parameters, the filtering estimate of the state of the system after the latest observation or prediction of the future evolution of the system. The usual method of choice for exact inference on discretely observed Markov jump processes (MJPs) on a finite or countably infinite state space is Bayesian inference via particle Markov chain Monte Carlo \cite[particle MCMC,][]{AndrieuDoucetHolenstein:2010} using a bootstrap particle filter \cite[e.g.][]{andrieu09,GolightlyWilkinson:2011,Wilkinson2012,mckinley2014,Owen:2015,koblents2015}. Other MCMC and SMC-based techniques are available \cite[e.g.][]{KYPRAIOS201742}, and, a further latent-variable-based MCMC method when the statespace is finite \cite[]{JMLR:v14:rao13a}.

Particle MCMC and SMC, however, are relatively complex algorithms, even more so when a bootstrap particle filter (simulation from the process itself) is not suitable and a bridge simulator is necessary, such as when observation noise is small or when there is considerable variability in the state from one observation to the next \cite[]{GW2015,GolightlySherlock2019}. 
In cases where the number of states, $d$, is finite, direct exact likelihood-based inference is available via the exponential of the infinitesimal generator for the continuous-time Markov chain, or rate matrix, $\matQ$. Whilst such inference is conceptually straightforward, it has been avoided in practice except in cases where the number of states is very small \cite[e.g.][]{AKTKJB19}. Matrix exponentiation has a computational cost of $\mathcal{O}(d^3)$, which, together with a lack of suitable tools in \texttt{R}, could explain the lack of uptake of this method. However, conceptually simple statistical inference via the matrix exponential is entirely practical in many cases even when the number of states is in the thousands or higher, for three main reasons:

\begin{enumerate}
\item Matrix exponentials themselves are never needed; only the \emph{product} of a vector and a matrix exponential is ever required.
\item The matrices to be exponentiated are infinitesimal generators and, as such, have a \emph{special structure}; furthermore, the vector that pre-multiplies the matrix exponential is non-negative.
\item The matrices to be exponentiated are \emph{sparse}; tools for basic operations with large, sparse matrices in \texttt{C++} and interfacing the resulting code with \texttt{R} have recently become widely available \cite[]{RcppArmadillo,ArmadilloSp}.
\end{enumerate}

The sparsity of $\matQ$ arises because the number of possible
‘next’ states given the current state is bounded by the number of reactions, which is typically
small.
This article describes matrix exponential algorithms suitable for statistical application in many cases, and demonstrates their use for inference, filtering and prediction. Associated code provides easy-to-use \texttt{R} interfaces to \texttt{C++} implementations of the algorithms, which are typically simpler and often faster than more generally applicable algorithms for matrix exponentiation.

Section \ref{sect.examples} describes the Susceptible-Infectious-Removed (SIR) model for the evolution of an infectious disease and the Moran model for the mixing of two alleles in a population, then briefly mentions many more such models where the statespace is finite, and a few where it is countably infinite. The two main examples will be used to benchmark and illustrate the techniques in this article.  As well as being directly of use for models with finite state spaces, exponentials of finite rate matrices can also be used to perform inference on Markov jump processes with a countably infinite statespace; see \cite{GHS2017} and \cite{SherlockGolightly2019}. The latter uses the uniformisation and scaling and squaring algorithms as described in this article, while the former uses the less efficient but more general algorithm of \cite{AMH11} (see Section \ref{sec.matexp}).

Section \ref{sec.data.and.likelihood} of this article presents the likelihood for discretely and partially observed data on a finite-statespace continuous-time Markov chain and presents two `tricks' specific to epidemic models, that allow for a massive reduction in the size of the generators that are needed compared with the size of the statespace. Section \ref{sec.matexp} describes the Matrix exponential algorithms and Section \ref{sec.numcomp} benchmarks some of the algorithms and demonstrates their use for inference, filtering and prediction. The article concludes in Section \ref{sec.discuss} with a discussion.

\subsection{Examples and motivation}
\label{sect.examples}
Both by way of motivation and because we shall use them later to illustrate our method, we now present two examples of continuous-time Markov processes, where a finite, sparse rate matrix contains all of the information about the dynamics.

For each Markov process, the set of possible states can be placed in one-to-one correspondance with a subset of the non-negative integers $\{1,\dots,d\}$. The off-diagonal elements of the rate matrix, $\matQ$, are all non-negative, and the $i$th diagonal element is $\matQ_{ii}=-\sum_{j=1,j\ne i}^d\matQ_{i,j}$. A chain that is currently in state $i$ leaves this state upon the first event of a Poisson process with a rate of $-\matQ_{i,i}$; the state to which it transitions is $j$ with a probability of $\matQ_{i,j}/(-\matQ_{i,i})$. Whilst the rate matrix, $\matQ$, is a natural description of the process, the likelihood for typical observation regimes involves the transition matrix, $\exp(\matQ t)$, the $(i,j)$th element of which is exactly
$\Prob{X_t=j|X_0=i}$.

\begin{example}
  \label{example.SIR} \textbf{The SIR model for epidemics}. The SIR model for a disease epidemic has $3$ species: those who are susceptible to the epidemic, $\sfS$, those both infected and infectious, $\sfI$, and those who have recovered from the epidemic and play no further part in the dynamics, $\sfR$. The non-negative counts of each species are denoted by $S$, $I$, and $R$. For relatively short epidemics the population, $n_{pop}$, is assumed to be fixed, and so the state of the Markov chain, represented by $(S,I)$, is subject to the additional constraint of $S+I\le n_{pop}$. The two  possible reactions and their associated rates are:
\[
\sfS+\sfI \stackrel{\beta SI}{\longrightarrow}2\sfI ,
~~~\mbox{and}~~~
    \sfI \stackrel{\gamma  I}{\longrightarrow}\sfR.
    \]
    
\end{example}

\begin{example}
  \label{example.Moran} \textbf{The Moran model for allele frequency} descibes the time evolution of the frequency of two alleles, $A_1$ and $A_2$ in a population with a fixed size of $n_{pop}$. Individuals with allele $A_1$ reproduce at a rate of $\alpha$, and those with $A_2$ reproduce at a rate of $\beta$. When an individual dies it is replaced by the offspring of a parent chosen uniformly at random from the whole population (including the individual that dies). The allele that the parent passes to the offspring usually matches its own, however as it is passed down an allele may mutate; allele $A_1$ switching to $A_2$ with a probability of $u$ and $A_2$ switching to $A_1$ with a probability of $v$. Let $\mathsf{A}_1$ and $\mathsf{A}_2$ represent individuals with alleles $A_1$ and $A_2$ respectively and let $N$ be the number of individuals with allele $A_1$. The two reactions are
  \[
  \mathsf{A}_1\stackrel{\lambda_N}{\longrightarrow}\mathsf{A}_2
~~~\mbox{and}~~~
    \mathsf{A}_2\stackrel{\mu_N}{\longrightarrow}\mathsf{A}_1.
    \]
    Setting $f_N=N/n_{pop}$, the corresponding infinitesimal rates are
    \[
    \lambda_N
    =
    \left(1-f_N\right)\left[\alpha f_N(1-u)+\beta(1-f_N)v\right]
~~~\mbox{and}~~~
    \mu_N
    =
    f_N\left[\beta (1-f_N)(1-v)+\alpha f_N u\right],
    \]
    where the unit of time is the expectation of the exponentially distributed time for an individual to die and be replaced.     \qed
\end{example}

The many other examples of interest include the SIS and SEIR models for epidemics \cite[e.g.][]{AndBritt2000}, dimerisation and the Michaelis-Menten reaction kinetics \cite[e.g.][]{Wilkinson2012}. Further examples but with an infinite statespace include the Schl\"ogel model \cite[e.g.][]{VellelaQian09}, the Lotka-Volterra predator-prey model \cite[e.g.][]{Wilkinson2012,drovandi16} and models for the autoregulation of the production of a protein \cite[e.g.][]{Wilkinson2012}, all of which are tackled using matrix exponentials in \cite{SherlockGolightly2019}.

\section{Data and likelihood calculations}
\label{sec.data.and.likelihood}
Denote the statespace of the Markov chain $\{X_t\}_{t\ge 0}$ by $\mathcal{X}=\{x^{(k)}\}_{k=1}^{d}$. Let the prior mass function across states be $\nu(x|\theta)$, the infinitesimal generator be $\matQ(\theta)$, and suppose there are observations $y_0,y_1,\dots,y_n$ at times $t_0,t_1,\dots,t_n$, where $Y_i|(X_i=x_i)$ has a mass function of $p(y_i|x_i,\theta)$, $i=0,\dots,n$.

\subsection{Likelihood for noisy and partially observed data}
%For simplicity of presentation we consider observations at regular time intervals, but it is straightforward to generalise this to irregular observations.

For any continuous-time Markov chain $\{X_t\}_{t\ge 0}$ with an infinitesimal generator, or rate matrix of $\matQ$, the $(x,x')th$ element of $\exp(\matQ t)$ gives the transition probability \cite[e.g.][]{Norris1997}:
\[
\Prob{X_t=x'\mid X_0=x}=\left[\exp(\matQ t)\right]_{x,x'},
\]
where here and elsewhere we abuse notation by identifying the state $x^{(i)}\in\mathcal{X}$ with the corresponding index $i\in\{1,\dots,d\}$.
  
Defining  the diagonal likelihood matrix to be 
$L_j(\theta)=\mathsf{diag}(p(y_j|x^{(1)},\theta),\dots,p(y_j|x^{(d)},\theta))$  
 and $\Delta_j=t_j-t_{j-1}$, $j=1,\dots,n$, the likelihood for the observations is then
\begin{align}
  \nonumber
  \Prob{y_0,\dots,y_n\mid \theta}
  &=
  \sum_{(x_0,\dots,x_n)\in\mathcal{X}^{n+1}}
  \Prob{X_0=x_0}\Prob{Y_0=y_0|X_{0}=x_0}\\
\nonumber
&~~~~~~~~~~~~~~~~
\prod_{j=1}^n\Prob{X_j=x_j| X_{j-1}=x_{j-1}}\Prob{Y_j=y_j|X_{j}=x_j}\\
  &=  \label{eqn.likelihood}
\nu(\theta)^\top L_0(\theta)\left[\prod_{j=1}^n \exp(\matQ(\theta) \Delta_j)L_j(\theta)\right]\underline{1},
\end{align}
where $\underline{1}$ is the $d$-vector of ones. Similarly, the filtering distribution after observation $y_m$ is
\begin{align}
  \label{eqn.filter}
  \Prob{X_{t_m}=x\mid y_0,\dots,y_m}
  &=
  \frac{
    \nu(\theta)^\top L_0(\theta)\left[\prod_{j=1}^m \exp(\matQ(\theta) \Delta_j)L_j(\theta)\right]}
  {
    \nu(\theta)^\top L_0(\theta)\left[\prod_{j=1}^m \exp(\matQ(\theta) \Delta_j)L_j(\theta)\right]\underline{1}}.
\end{align}
Consider the required multiplication from left to right: since the likelihood vectors $L_j(\theta)$ are diagonal, pre-multiplication by a $d$-vector is an $\mathcal{O}(d)$ operation. Pre-multiplcation of the exponential of a sparse matrix by a $d$-vector via the uniformisation algorithm is also $\mathcal{O}(d)$ (see Section \ref{sec.uniformisation}), so the entire likelihood calculation is $\mathcal{O}(d)$. In the case of certain epidemic models $d$ itself can be much smaller than might naively be assumed.

\subsection{Statespace reduction for epidemic models}
\label{sec.statespace.red}
\cite{HCS2018} points out that the evolution of the Markov chain in an SIR model from one observation time to the next can be described entirely by the number of new infections and the number of new removals, $B_I$ and $B_R$, neither of which can be negative. Consider the case of exact observations and suppose, for example, that in a population of size $n_{pop}=500$, $x_a=(S_a,I_a,R_a)=(485,2,13)$ and for some $t>0$, $x_{a+t}=(470,3,27)$. Then $b_R=R_{a+t}-R_a=14$
and $b_I=S_{a}-S_{a+t}=15$. The size of the statespace for evolution between time $a$ and time $a+t$, $\mathcal{X}_{a}^{a+t}$, is then reduced from the size of the full statespace, $(n_{pop}+1)(n_{pop}+2)/2=125751$ to $(b_I+1)(b_R+1)=240$.

In \cite{HCS2018}, a recursive formula for the Laplace transform of the transition probability to a given new state in terms of transition probabilities for old states then permits estimation of the transition vector from a known initial starting point in $\mathcal{O}(d)$ operations, where $d$ is the dimension of the statespace actually required.

We may use the same statespace formulation as \cite{HCS2018}, provided we include an additional coffin state, $\Coffin$, with $Q_{\Coffin,x}=0$ for all $x\in \mathcal{X}_a^{a+t}\cup \Coffin$. Any births that would leave the statespace (and hence contradict the observation at time $a+t$) instead go to $\Coffin$.
We also provide a further reduction in the size of the statespace, by a factor of up to one half. The current number of infections can never be negative, so, throughout the time interval $[a,a+t]$, $b_R\le I_a+b_I$. In the example above, this reduces the statespace size still further, from $240$ to $162$.

Although we do not examine it here, a similar reduction of the statespace but for the SEIR model, to $(b_E,b_I,b_R)$, is described in \cite{HCS2018}. A further reduction, by  a factor of up to $6$, is possible by observing that neither $E_{a+t}=E_a+b_E-b_I$ nor $I_{a+t}=I_a+b_{I}-b_R$ can be negative.

In epidemics, typically only removals are observed. Thus, to take advantage of the reduced statespace formulation, latent variables representing the number infected at each observation time must be introduced. This makes Bayesian inference via MCMC feasible for relatively large populations. For example, \cite{HCS2018} performs inference for the SIR model using data from the Ebola outbreak in regions of Guinea, and in Section \ref{sec.measles} we perform inference on data from the 2013 Measles outbreak in Swansea, Wales.

\section{Matrix exponentiation}
\label{sec.matexp}
The exponential of a $d\times d$ square matrix, $\matM$ is defined via its infinite series: $e^\matM=\sum_{i=0}^\infty \frac{1}{i!}\matM^i$.
As might be anticipated from the definition, for a $d\times d$ matrix, algorithms for evaluating $\exp {(\matM)}$ take $\mathcal{O}(d^3)$ operations \cite[see][for a review of many such methods]{Dubious}. However, for a $d$-vector, $v$, the product $\exp{(Mt)}~v$ is the solution to the initial value problem $w(0)=v$, $dw/dt=\matM w$, and is the key component of the solution to more complex differential equations such as $dw/dt=\matM w +B u(t)$. For this reason the numerical evaluation of the action of a matrix exponential on a vector has received considerable attention of itself \cite[e.g.][]{GS92,Saad92,Sidje98,AMH11}.

When $\matM$ is dense, 
\begin{equation}
  \label{eqn.matexponvec}
  \exp(\matM)~v=\sum_{i=0}^\infty \frac{1}{i!}\matM^iv
\end{equation}
can be evaluated in $\mathcal{O}(d^2)$ operations if the series is truncated at an appropriate point. However, motivated by the examples in Section \ref{sect.examples} our interest lies in large \emph{sparse} matrices, and the number of operations can then be reduced to $\mathcal{O}(rd)$, where $r$ is the average number of entries in each row of $\matM$.

With double-precision arithmetic, real numbers are stored to an accuracy of approximately $10^{-16}$. Thus, evaluation of the exponential of a large negative number via its Taylor series is prone to potentially enormous round-off errors due to the almost cancellation of successive large positive and negative terms; a similar problem can affect the exponentiation of a matrix. Such issues are typically circumvented via the identity
\begin{align}
  \label{eqn.use.K}
\exp(\matM)v &= \left[\prod_{k=1}^K \exp(\matM/K)\right] v,
\end{align}
applied for a sufficiently large integer $K$, and evaluated via $K$ successive evaluations of product of $\exp(\matM/k)$ and a vector.
The calculation on the right of \eqref{eqn.use.K} typically involves many more numerical operations than the direct calculation on the right of \eqref{eqn.matexponvec}, so $K$ should be the \emph{smallest} integer that leads to the required precision by mitigating sufficiently against the cancellation of large positive and negative terms. This minimises both the accumulation of rounding errors and the total compute time given the required accuracy.

One common technique for such multiplication, exemplified in the popular \texttt{Expokit} \texttt{FORTRAN} routines \cite[]{Sidje98}, estimates $e^{\matM/K}v$
via its projection on to the Krylov subspace of  
$\mathsf{Span}\{v, \matM v,\dots,\matM^{n-1}v\}$,
where $n<<d$. A second method is provided in \cite{AMH11}, where the key contributions lie in the method for choosing $K$ and for choosing a suitable truncation point for the infinite series, as well as a means of truncating each series early depending on the behaviour of recent terms.

Both of the above algorithms use the fact that $\matM$ is sparse and that only the action of $\exp(\matM)$ on a vector is required, but neither uses the special structure of the problem of interest to us: we require
$\nu^\top\exp(\matQ t)$ where $\matQ$ is a rate matrix and $\nu$ is a non-negative vector. Since $\matQ t$ is also a rate matrix, we henceforth set $t=1$ without loss of generality.
Let
\begin{equation}
\label{eqn.define.rho.P}
  \rho:=\max_{i=1,\dots,d} |\matQ_{ii}|
  ~~~\mbox{and}~~~
  \matP=(1/\rho)\matQ+I.
\end{equation}
$\matP$ is a Markov transition matrix, and  the key observation is that 
\begin{equation}
  \label{eqn.key.reln.P.Q}
\exp \matQ=\exp (\rho \matP-\rho \matI)
=\exp(-\rho)\exp(\rho \matP)
=\sum_{i=0}^\infty\exp(-\rho)\frac{\rho^i}{i!}\matP^i.
\end{equation}
Firstly, $\matP$ has no negative entries so cancellation of terms with alternating signs is no longer a concern. Secondly, $\exp \matQ$ can be interpreted as a mixture over a $\mbox{Poisson}(\rho)$ random variable $I$, of $I$ transitions of the discrete-time Markov chain with a transition matrix of $\matP$. 

The next two subsections detail variations on two existing algorithms that utilise this special structure: the \emph{uniformisation} algorithm and a variation on the \emph{scaling and squaring} algorithm. For sparse rate matrices, the uniformisation algorithm has a cost of $\mathcal{O}(\rho d)$, whereas the scaling and squaring algorithm has a cost of
$\mathcal{O}(d^3\log \rho)$. Thus, the uniformisation algorithm is preferred when $\rho$ is small, and scaling and squaring when $\rho$ is large but $d$ is relatively small. We now describe the two algorithms in detail.

\subsection{The uniformisation algorithm}
\label{sec.uniformisation}

In many statistical applications, the most appropriate algorithm  for calculating $\mu^\top:=\nu^\top \exp \matQ $ is the uniformisation algorithm \cite[e.g.][]{ReibmanTrivedi1988,SidjeStewart1998}.
This estimates $\mu^\top$ by truncating a single series none of whose terms can be negative, rather than truncating multiple series where terms may change sign as in \cite{AMH11}. Given an $\epsilon>0$, the algorithm calculates an approximation, $\muhat$, to $\mu$ by picking a truncation point for the infinite series, such that, if $\nu$ were a probability vector, the (guaranteed to be non-negative) amount of \emph{true} missing probability over all of the $d$ dimensions is controlled:
\[
0< 1-\frac{||\muhat^*||_1}{||\nu||_1} < \epsilon,
\]
where $\muhat^*$ is the probability vector that would be calculated if there were no rounding errors, and the only errors were due to the truncation of the infinite series.
Typically we aim for $\epsilon$ to be similar to the machine's precision. We control the absolute truncation error and note that with any truncation of the power series, it is impossible to obtain general control of the \emph{relative error in a given component} of $\mu$, $|\muhat_i/\mu_i-1|$. Consider, for example, a Moran process (Example \ref{example.Moran}), where $\matQ$ is tridiagonal. Then $\matQ^k$ is also banded, with a band width of $2k+1$. For any given $m_{max}$, and $\nu=(1,0,0,\dots)$, set $d>m_{max}+1$. The truncated approximation to $e^{\matQ}$ gives a transition probability of $0$ for all states above $m_{max}+1$, yet, in truth there is a non-zero probability of such a transition.

From \eqref{eqn.key.reln.P.Q},
\[
\mu^\top=\nu^\top e^Q = e^{-\rho}\nu^\top \sum_{i=0}^\infty \frac{\rho^i}{i!}P^i
\approx
e^{-\rho}
\sum_{i=0}^{m} \frac{\rho^i}{i!}\nu^\top P^i=:\muhat^{*\top}.
\]
Now,
\[
\sum_{i=1}^d \muhat^*_i = \muhat^{*\top}1  = e^{-\rho}\sum_{i=0}^m \frac{\rho^i}{i!}\nu^\top P^i 1
=
||\nu||_1 e^{-\rho}\sum_{i=0}^m \frac{\rho^i}{i!}.
\]
So the absolute relative error, or (when $\nu$ is a probability vector) missing probability mass, due to truncation is
\[
r_m(\rho):=e^{-\rho}\sum_{i=m+1}^\infty \frac{\rho^i}{i!},
\]
the tail probability of a $\mathsf{Poisson}(\rho)$ random variable.
Of direct interest to us is 
  \[m_{\epsilon}(\rho)
  :=
  \inf\{m\in \mathbb{N}: r_m(\rho)\le \epsilon\},
  \]
  the smallest $m$ required to achieve an error of at most $\epsilon$, or, essentially, the quantile function for a $\mathsf{Poisson}(\rho)$ random variable, evaluated at $1-\epsilon$. Chebyshev's inequality applied to $X/\rho$, where $X\sim \mathsf{Poisson}(\rho)$ gives $\Prob{|X/\rho - 1|\ge 1/\sqrt{\epsilon \rho}}\le \epsilon$, implying the $m=\mathcal{O}(\rho)$ computational cost given earlier in this section.
  
  In many programming languages, standard functions are available to evaluate $m_{\epsilon}(\rho)$. However, for example, in \texttt{R} we find
\begin{verbatim}
> rho=100; eps=1e-16
> qpois(eps,rho,lower.tail=FALSE)
[1] Inf
> ppois(193,rho,lower.tail=FALSE) # 193 is correct answer, not infinity
[1] 5.713551e-17
> eps=1e-15
> qpois(eps,rho,lower.tail=FALSE)
[1] 185
> ppois(185,rho,lower.tail=FALSE)
[1] 1.035777e-14
> ppois(189,rho,lower.tail=FALSE) # 189 is correct answer, not 185
[1] 8.017165e-16  
\end{verbatim}
\emph{i.e.}, an inability to calculate $m_{\epsilon}(\rho)$ correctly given the small $\epsilon$ values that we require; the underlying functions are also callable from \texttt
{C++} and lead to the same error. In Appendix \ref{sec.mepsrho} we provide sharp bounds on $m_{\epsilon}(\rho)$, and this leads to an accurate methodology for its exact calculation, producing the same (correct) answers as the \texttt{C++} \texttt{boost} library (which we have not been able to use with \texttt{RCpp}) and up to twice as quickly.

The uniformisation algorithm is presented as Algorithm \ref{alg.unif}.
For large values of $\rho$, although there is no problem with large positive and negative terms cancelling, it is possible that the partial sum
$\sum_{i=0}^k \frac{\rho^i}{i!}$ might exceed the largest floating point number storable on the machine. We circumvent this problem by occasionally renormalising the vector partial sum when the most recent contribution is large, and compensating for this at the end; see lines \ref{renormaa}, \ref{renormb} and \ref{renormc}.
%One could, alternatively, renormalise when the partial sum itself becomes large.

\begin{algorithm}
  \label{alg.unif}
  \caption{Uniformisation algorithm for $\nu^\top e^{\matQ}$ with a missing mass of at most $\epsilon$.}
  \begin{algorithmic}[1]
    \State $\rho\gets \max_{i=1}^d |Q_{i,i}|$;
    $\matM\gets \matQ+\rho \matI_d$; $BIG\gets 10^{100}$.
    \State Find $m_{\epsilon}(\rho)$.
    \State $b\gets||\nu||_1$; $c \gets 0$. 
    \If{$b>BIG$} 
    \State  $\nu\gets \nu/b$; $c\gets c+\log b$; $b\gets 1$. \label{renormaa}
    \EndIf
    \State $v_{pro}\gets v_{sum}\gets \nu$.
    \State $f \gets 1$.
    
       \For{$j$ from $1$ to $m$}
       \State $v_{pro}^{\top}\gets v_{pro}^\top \matM/f$; $b\gets b\rho/f$.
       \State $v_{sum}\gets v_{sum}+v_{pro}$. \label{accum}
       \If{$b>BIG$}
       \State $v_{pro}\gets v_{pro}/b$; $v_{sum}\gets v_{sum}/b$; $c\gets c+\log b$; $b\gets 1$. \label{renormb}
       \EndIf
       \State $f\gets f+1$.
       \EndFor
       \State \textbf{return} $e^{c-\rho}\times v_{sum}$. \label{renormc}
  \end{algorithmic}
\end{algorithm}

\subsection{Scaling and squaring}
One of the simplest, yet most robust methods for exponentiating any square matrix is the scaling and squaring algorithm \cite[e.g.][]{Dubious}. When the square matrix is an infinitesimal generator, this method can be made even more robust using the reformulation in \eqref{eqn.key.reln.P.Q}. Furthermore, when not $\exp \matQ$ but $\nu^\top\exp \matQ$ is required, some further computational savings can be obtained.

The basic scaling and squaring algorithm takes advantage of the identity
\[
\exp(\matM)=\left[\exp(\matM/2^s)\right]^{2^s},
\]
where for any integer $s$, a square matrix is raised to the power of $2^s$ by squaring it $s$ times. We set $\matM=\matQ+\rho\matI=\rho\matP$ from \eqref{eqn.define.rho.P}. And define $\matM_{small}=\matM/2^s$. First, $\exp(\matM_{small})$ is approximated via the uniformisation algorithm applied to a matrix \cite[e.g.][]{Ross1996}: $\sum_{i=0}^m \matM_{small}^i/i!$. This quantity is then squared $s$ times. A linear search between tight upper and lower bounds gives the $\widehat{s}$ which minimises the total number of matrix multiplications, $m_\epsilon(\rho/2^s)+s$; we set $s\leftarrow \widehat{s}-\min(2,\widehat{s})$ as a compromise for the fact that when $\matQ$ is sparse, the individual matrix multiplications used to evaluate $\exp(\matM_{small})$ are cheaper than those used to square it, since $\exp(\matM_{small})$ is dense.

When evaluating $\nu^\top \exp(\matQ)=\exp(-\rho)\nu^\top \exp(\matM)$ via scaling and squaring with $s>0$ it is never most efficient to first evaluate $\exp(\matM)$. Let $s_1$ and $s_2$ be two integers such that $s_1+s_2=s$. Then
\[
\nu^\top \exp(\matM)
=
\nu^\top [\exp(\matM_{small})]^{2^{s_1}}[\exp(\matM_{small})]^{2^{s_1}}\dots [\exp(\matM_{small})]^{2^{s_1}},
\]
with $2^{s_2}$ matrix vector products. The cost of $s_1$ matrix squares and $2^{s_2}$ vector-matrix products (where the matrix is dense) is $s_1d^3+2^{s_2}d^2$. We round the minimiser down to the nearest integer for simplicity, setting
\begin{equation}
  \label{eqn.ssub}
s_2=\min\left(s,\lfloor (\log d-\log\log 2) /\log 2\rfloor\right)
\end{equation}
Even with $d=2$ this gives $s_2=\min(s,1)$.

\begin{algorithm}
  \label{alg.SS}
  \caption{Scaling and squaring algorithm for $\nu^\top e^{\matQ}$ with a missing mass of at most $\epsilon$.}
  \begin{algorithmic}[1]
    \State $\rho\gets \max_{i=1}^d |Q_{i,i}|$.
    \State Find $s$ via linear search; $\rho_{small}\gets \rho/2^s$; find $m_{\epsilon}(\rho_{small})$; find $(s_1,s_2)$ via \eqref{eqn.ssub}.
    \State $\matM_{small}\gets (\matQ+\rho \matI)/2^s$.
    \State $\nu_{pro}\gets \nu$.
    \State $\matA_{pro}\gets  \matM_{small}$; $\matA_{sum} \gets \matI+\matM_{small}$
    \State $f \gets 2$.
    
       \For{$j$ from $2$ to $m$}
       \State $\matA_{pro}\gets \matA_{pro} \matM_{small}/f$.
       \State $\matA_{sum}\gets \matA_{sum}+\matA_{pro}$. 
       \State $f\gets f+1$.
       \EndFor
       \State $\matA_{sum}\gets e^{-\rho_{small}} \matA_{sum}$
       \For{$j$ from $1$ to $s_1$}
       \State $\matA_{sum}\gets \matA_{sum}\times \matA_{sum}$.
       \EndFor
       \For{$j$ from $1$ to $2^{s_2}$}
       \State $\nu_{pro}^\top\gets \nu_{pro}^\top \matA_{sum}$.
       \EndFor
       \State \textbf{return} $\nu_{pro}^\top$. 
  \end{algorithmic}
\end{algorithm}

\subsection{Improvements}
\label{sect.options}
We now describe two optional extensions: renormalisation, which improves the accuracy of any matrix exponentiation algorithm used on a rate matrix, and two-tailed truncation, which is unique to the uniformisation algorithm and allows a small computational saving.

Since $a:=\sum_{i=1}^d \mu_{i}=\sum_{i=1}^d \nu_{i}$ there is no need to keep track of the logarithmic offset ($c$ in Algorithm \ref{alg.unif}). Instead the final vector ($v_{sum}$ in Algorithm \ref{alg.unif}) is renormalised at the end so that its components sum to $a$. 

Two-tailed truncation \cite[e.g.][]{ReibmanTrivedi1988} permits a small reduction in the computational cost of the uniformisation algorithm with no loss of accuracy. When $\rho$ is moderate or large, the total mass of probability from the initial value of $v_{sum}$ and the early values accumulated into $v_{sum}$ (Steps 6 and 10 of Algorithm \ref{alg.unif}) is negligible (has a relative value smaller than $\epsilon/2$, say) compared with the sum of the later values. In such cases $v_{sum}$ may be initialised to $0$ and step \ref{accum} omitted for values of $j$ beneath some $m_{lo}$. Proposition \ref{prop.lopsided} below shows that if $m$ is chosen such that $\Prob{\mathsf{Po}(\rho)>m}\le \epsilon/2$ then setting
$m_{lo}:=\max(0,2\lfloor \rho-0.5\rfloor-m)$ ensures that the missing probability mass is no more than $\epsilon$. For large $\rho$,
$m-m_{lo}=\mathcal{O}(\sqrt{\rho})$, so with two-tailed truncation the cumulative cost of Step \ref{accum} dwindles compared with the other $\mathcal{O}(d)$ costs, which are repeated $\mathcal{O}(\rho)$ times.

  \begin{proposition}
    \label{prop.lopsided} Given $\rho>0$, let $p_n=e^{-\rho}\rho^n/n!=\Prob{\mathsf{Poisson(\rho)}=n}$, and let $c=\lfloor \rho-1/2\rfloor$. Then for $a\le c-1$,
    \[
\sum_{j=0}^{c-a-1}p_j < \sum_{j=c+a+1}^\infty p_j.
    \]
\end{proposition}
  \begin{proof}
    For any integer $b$, and $1\le i \le b$,
    \[
    \frac{p_{b-i}}{p_{b+i}}=\rho^{-2i}b(b+1)(b-1)(b+2)\dots (b-i+1)(b+i)
    =
    \rho^{-2i}\left[b_*^2-\frac{1}{2^2}\right]
    %\left[b_*^2-\frac{3^2}{2^2}\right]
    \cdots \left[b_*^2-\frac{(2i-1)^2}{2^2}\right] %b_*^2-\frac{(2i-1)^2}{2^2}\right],
    \]
    where $b_*=b+1/2$. Hence, if $b_*\le \rho$, $p_{b-1}/p_{b+i}<1$, and so
    \[
    \sum_{j=0}^{\lfloor b^*\rfloor -a-1} p_i=
\sum_{i=a+1}^{\lfloor b_*\rfloor} p_{b-i}< \sum_{i=a+1}^{\lfloor b_*\rfloor} p_{b+i}<\sum_{i=a+1}^{\infty}p_{b+i}.
\]
  \end{proof}

  \subsection{Implementation}
Our \texttt{C++} implementation uses the recent basic sparse matrix functionality in the \texttt{C++} \texttt{Armadillo} library \cite[]{Armadillo,ArmadilloSp} to calculate $\nu^\top \exp \matQ$, where $\nu$ is non-negative and $\matQ$ is a large, sparse rate matrix. Direct function calls from the \texttt{R} programming language are enabled through \texttt{RcppArmadillo}  \cite[]{RcppArmadillo}. For completeness, the functions can also be called with dense rate matrices. The functions are collected into the \texttt{expQ} package which is downloadable from \url{https://github.com/ChrisGSherlock/expQ} and are briefly outlined in Appendix \ref{sec.expQ}. 

The speed of a vector multiplication by a sparse-matrix depends on the implementation of the sparse matrix algorithm. In \texttt{R} \cite[]{Rcite} and in \texttt{C++} \texttt{Armadillo}, sparse matrices are stored in column-major order. Hence pre-multiplication of the sparse matrix by a vector, $\nu^\top Q$, is much quicker than post multiplcation, $Q\nu$. In other languages, such as \texttt{Matlab}, sparse matrices are stored in row-major order and post-multiplication is the quicker operation, so $Q^\top$ should be stored and used, rather than $Q$.

\section{Numerical comparisons and demonstrations}
\label{sec.numcomp}
In \cite{AMH11} their new algorithm (henceforth referred to as AMH) is compared across many examples against state-of-the-art competitors, including, in particular, the \texttt{expokit} function \texttt{expv} \cite[]{Sidje98}. In most of the experiments AMH is found to give comparable or superior accuracy together with superior computational speed.
Given these existing comparisons and that the superiority of the uniformisation algorithm over the algorithm of \cite{AMH11} (for rate matrices) is not the main thrust of this paper, we perform a short comparison of accuracy and speed for two different likelihood calculations for an SIR model fitted to data from the Eyam plague. We compare our implementation of the uniformisation algorithm, the algorithm of AMH, the \texttt{expAtv} function which is from the \texttt{R} package \texttt{expm} and uses the method of \cite{Sidje98}, and the bespoke algorithm for epidemic processes in \cite{HCS2018}. Since it would be unfair to compare the clock-speeds for the \texttt{Matlab} code for AMH  directly with those of our \texttt{RCppArmadillo} implementation, we compare the number of sparse vector-matrix multiplications that are required.

The highest accuracy available in \texttt{C++} using sparse matrices and the \texttt{armadillo} linear algebra library is double precision, which we used throughout in our implementation of both of our algorithms. For the uniformisation and scaling and squaring algorithms we used $\epsilon=10^{-15}$, and for AMH we used the double-precision option. For \texttt{expAtv} and for \cite{HCS2018} we use the default package setting.

\subsection{Comparison with other matrix exponentiation algorithms}
\label{sec.cmp.HCSexpv}
To examine the speed and accuracy of the algorithm 
we consider the collection (see the first three rows of Table \ref{table.Eyam}) of $(S,I)$ (susceptible and infected) values, which originated in \cite{Raggett82} and were used in \cite{HCS2018}, for the Eyam plague. We set the parameters to their maximum-likelihood estimates, $(\beta,\gamma)=(0.0196,3.204)$ and consider the likelihood for the data in Table \ref{table.Eyam}. In addition, 
to mimic the size of potential changes between observation times and the size of the elements of the rate matrix from a larger population, we also evaluated the likelihood for the jump directly from the data at time $0$ to the data at time $4$. The final three rows of Table \ref{table.Eyam} refer to the rate matrix for the transition between consecutive observations and provide the dimension the matrix first using the reformulation of \cite{HCS2018} and then applying the improvement described in Section \ref{sec.statespace.red}; the final row is the absolute value of the largest entry of $\matQ$, $\rho$. The rate matrix for the single jump between times $0$ and $4$ had $d_{HCS}=30789$, $d=16082$ and $\rho\approx 3439.5$. The full statespace has a size of $34453$. Thus, for large changes, the main reduction in size arises from the improvement in Section \ref{sec.statespace.red}, but for small jumps this provides a smaller relative reduction compared with that in \cite{HCS2018}.

\begin{table}
\begin{center}
  \caption{\small{Time (in units of 31 days), and numbers of susceptibles and infecteds, originally from \cite{Raggett82}. The final rows indicates, for each pair of consecutive observations, the size of the statespace for evaluating the transition probability and the $\rho$ value for the associated rate matrix.} \label{table.Eyam}}
\begin{tabular}{c|rrrrrrrr}
  Time&0&0.5&1.0&1.5&2.0&2.5&3.0&4.0\\
  \hline
  S   &254&235&201&153&121&110&97&83\\
  I   &7&14&22&29&20&8&8&0\\
  \hline
  $d_{HCS}$ &-&261&946&2059&1387&289&197&346\\
  $d$ &-&245&867&1868&1308&282&181&240\\
  $\rho$ &-&101.5&171.4&217.1&170.1&83.1&53.6&106.3\\
\end{tabular}
\end{center}
\end{table}
For the uniformisation and scaling and squaring algorithm, with $\epsilon=10^{-15}$, the algorithm of \cite{HCS2018} and the \texttt{expAtv} function from the \texttt{R} package \texttt{expm} package \cite[]{Sidje98} we found the CPU time for $1000$ estimations of the likelihood ($20$ estimates for the likelihood for the jump from $t=0$ to $t=4$). We also recorded the  error in the evaluation of the log likelihood. Since for uniformisation, using renormalisation and two-tailed truncation together produced the fastest and most accurate evaluations, we only considered this combination. Given that the true likelihood is not known, the error using uniformisation, from scaling and squaring and from \cite{AMH11} were approximately bounded by examining their discrepancy from each other.
The results are presented in Table \ref{table.Eyam.res}.

\begin{table}
    \begin{center}
  \caption{\small{Timings for estimating the full log-likelihood ($1000$ repeats) and the log-likelihood for the jump from the initial to the final observation ($20$ repeats) for the Eyam data set, number of sparse vector-matrix multiplications for one repeat,  and the accuracies of the estimates. Results are given for the method of \cite{HCS2018} (HCS), the \texttt{expAtv} function in the \texttt{expm} package, which uses the Krylov subspace techniques of \cite{Sidje98}, the method of \cite{AMH11} (AMH), the uniformisation algorithm (Unif) and the scaling and squaring algorithm (SS). $^1$ The timing for SS on the jump likelihood was estimated from a single repeat. \label{table.Eyam.res}}}
  
  \begin{tabular}{c|ccc|ccc}
    &\multicolumn{3}{c|}{Full likelihood}&\multicolumn{3}{c}{Jump likelihood}\\
  Algorithm&Time (secs)&Mult&Accuracy & Time (secs) &Mult& Accuracy\\
  \hline
  HCS & 45.3 &-& $5.7\times 10^{-8}$ & 9.7 &-&  $4.3\times 10^{-9}$\\
  expAtv & 558.5 &-& $1.6\times 10^{-10}$& 323.2 &-& $8.2\times 10^{-11}$\\
  AMH &-& 3701 &$<1\times 10^{-15}$&-& 14300 & $<4\times 10^{-14}$\\
  Unif  & 18.72 &1596& $<1\times 10^{-15}$ &15.2 &3921&$<6\times 10^{-14}$\\
  SS  &1678 &-& $1.1\times 10^{-13}$ &8940$^1$&-&$<6\times 10^{-14}$\\
  \end{tabular}
%  \begin{tabular}{c|cc|cc}
%    &\multicolumn{2}{c|}{Full likelihood}&\multicolumn{2}{c}{Jump likelihood}\\
%  Algorithm&Time (secs)&Accuracy & Time (secs) & Accuracy\\
%  HCS&47.98--48.00&$5.7\times 10^{-8}$&48.29--48.42&$4.3\times 10^{-9}$\\
%  expv&37.02--37.05$^*$& $1.2\times 10^{-13}$&129.40--130.53$^*$&$4.7\times 10^{-12}$\\
%  SPS($10^{-9}$)&16.57--16.61&$8.5\times 10^{-9}$&90.82--91.03&$1.3\times 10^{-9}$\\
%  SPSr($10^{-9}$)&16.62--16.66&$3.4\times 10^{-9}$&90.86--90.95&$3.3\times 10^{-10}$\\
%  SPS2($10^{-9}$)&16.47--16.52&$5.1\times 10^{-9}$&84.92--85.03&$6.9\times 10^{-10}$\\
%    SPS2r($10^{-9}$)&16.44--16.51&$2.6\times 10^{-9}$&85.05--85.24&$1.2\times 10^{-10}$\\
%  SPS2r($10^{-16}$)&19.77--19.80&$<1\times 10^{-15}$&96.54--96.57&$<1\times 10^{-14}$\\  
%  \end{tabular}
    \end{center}
\end{table}

Scaling and squaring is extremely slow in these high-dimensional scenarios; however, \cite{SherlockGolightly2019} provides a bistable example, the Schl\"ogel model, where $d\approx 100-200$ but $\rho>10^5$, and the scaling and squaring algorithm outperforms uniformisation by orders of magnitude.

Since $m=\mathcal{O}(\rho)$ the choice of tolerance, $\epsilon$, typically has only a small effect on the speed of the uniformisation algorithm. For the full likelihood evaluation, uniformisation is over twice as fast as the algorithm of \cite{HCS2018} and approximately thirty times as fast as \texttt{expAtv}, and is more accurate than either; it is also over twice as fast as the algorithm of \cite{AMH11}, although both are very accurate.

For the single large jump between observations, we see the same pattern in terms of accuracy. There is a gain in efficiency by using two-tailed-truncation because $\rho$ is larger ($m_{lo}=3081$ and $m=3797$), but despite this, the method of \cite{HCS2018} is now more efficient than uniformisation, although considerably less accurate than it. Again, \texttt{expAtv} is over twenty times slower than uniformisation and less accurate, and AMH is over three times slower than uniformisation.

\subsection{Maximum likelihood inference, filtering and prediction}
\label{sec.Moran}
We now consider the Moran model, which has four unknown parameters: $(\alpha,\beta,u,v)$ and $n_{pop}=1000$. Setting $(\alpha,\beta,u,v)=(1,0.3,0.2,0.1)$, we simulate a path of the process for $T=10000$ time units. We then sample $51$ observations at times $0,200,400,\dots,10000$, by taking the value of the process at each of these times and adding independent noise with a distribution of $\mathsf{Bin}(800,0.5)-400$.

We then perform inference on $\theta=(\log \alpha, \log \beta, \log [u/(1-u)], \log [v/(1-v)])$ by maximising the likelihood based on all the data and, separately, based on the data up to $T=5000$. In each of these two data scenarios we find the filtering distribution, $\Prob{X_T|y_{0:T},\thetahat}$, at time $T$ via \eqref{eqn.filter}; finally we predict forward from $T$ in steps of $200$ for a further time of $T_{pred}=5000$  by repeatedly multiplying the current distribution vector by $\exp(200 \matQ(\thetahat))$. The true values, observations and filtering and prediction distributions are shown in Figure \ref{fig.Moran}. The whole process of inference and prediction took less than two minutes on a single i7-3770 CPU running at $3.40$GHz. Further, after defining $\matQ$, only $10$ lines of $\texttt{R}$ code are required to calculate the log-likelihood, and fewer than this to produce the filtering distribution (see Appendix \ref{app.code}).

\begin{figure}
\begin{center}
  \includegraphics[scale=0.44,angle=0]{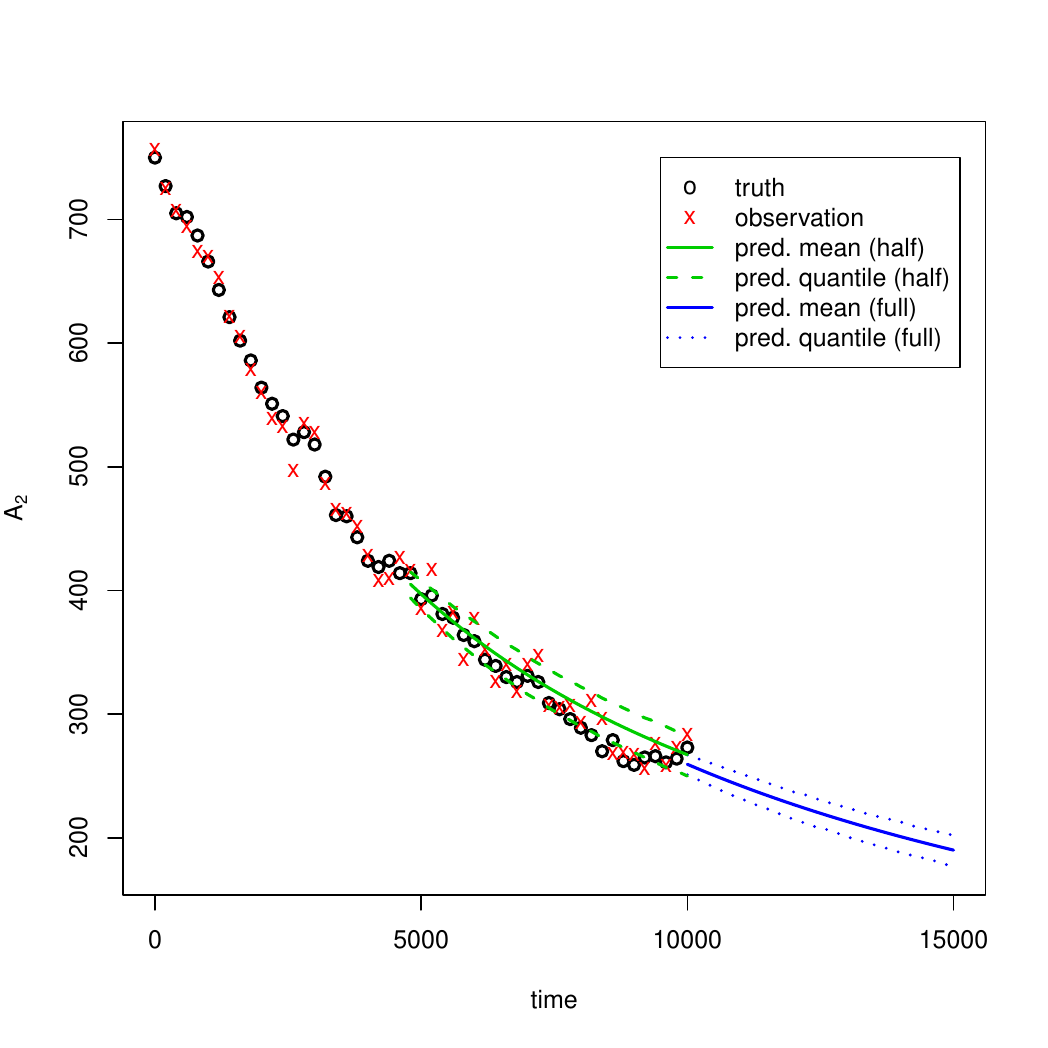}
  \includegraphics[scale=0.44,angle=0]{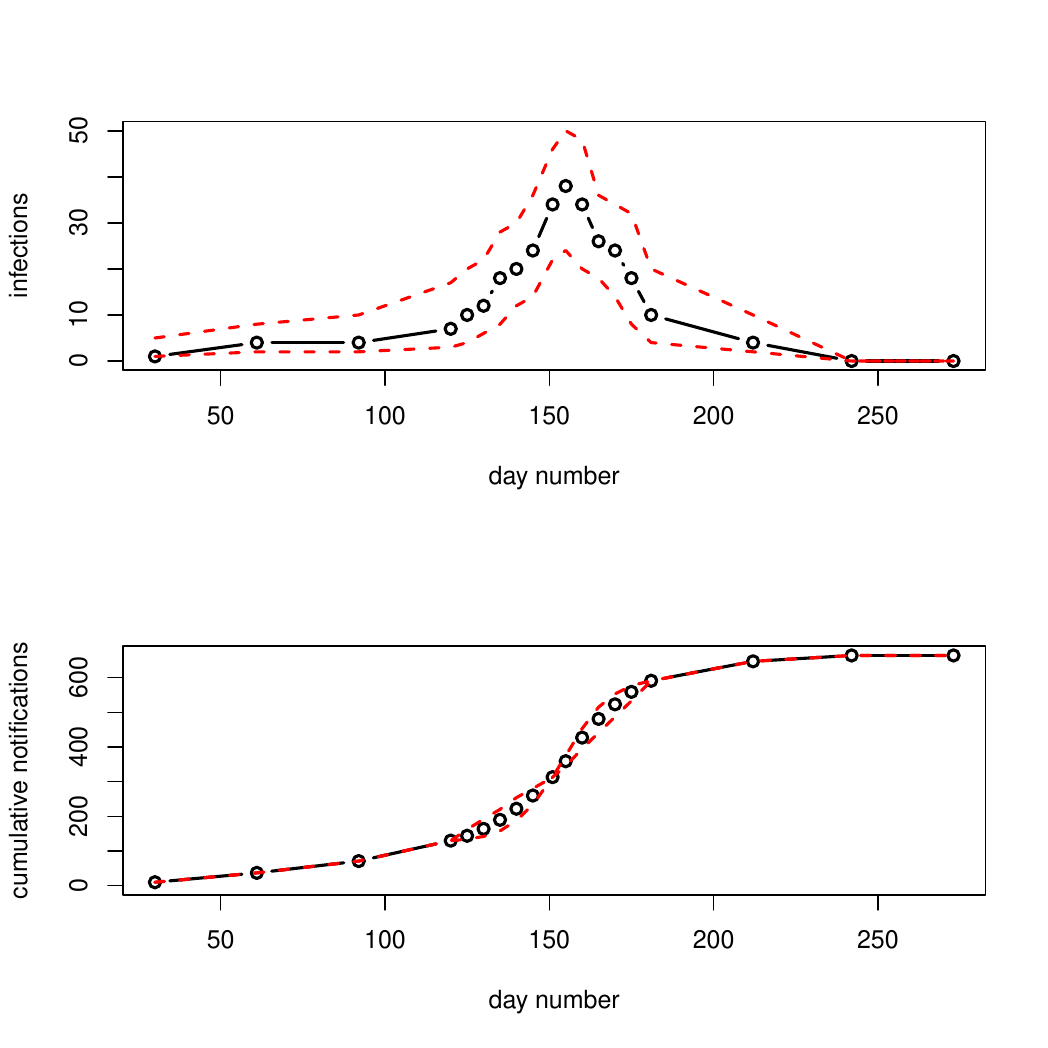}
  \caption{\small{Moran model (left): true values ($\mathsf{o}$), observations ($\mathsf{x}$), filtering/prediction mean (solid lines) and $95\%$ quantiles (dashed and dotted lines) for a further time of $5000$ from data up to $T=5000$ and data up to $T=10000$. Swansea measles SIR model (right): posterior median (solid line, with $\mathsf{o}$ to show the positions) and posterior $95\%$ quantiles for the number of infected people at each real or latent observation time (top) and the cumulative number recovered by that time (bottom).}
\label{fig.Moran}
}
\end{center}
\end{figure}

\subsection{Bayesian inference for the Swansea measles epidemic of 2013}
\label{sec.measles}
The largest measles outbreak in the United Kingdom between 2011 and 2019 centred around Swansea, Wales and occurred between November 2012 and July 2013. Of the $1219$ cases in mid- and west-Wales, $664$ occurred in the Swansea Local Authority (LA) area, $243$ in the nearby Neath and Port Talbot LA and fewer than $80$ occurred in any of the other individual LA areas in West or South Wales (\url{http://www.wales.nhs.uk/sitesplus/888/page/66389}, accessed February 10th 2020). A reduction in uptake of the MMR (Measles, Mumps, Rubella) vaccine has been blamed \cite[e.g.][]{Measles2013A} for this, with particularly low rates reported in Swansea (\url{https://en.wikipedia.org/wiki/2013_Swansea_measles_epidemic}, accessed February 10th 2020). 

The basic reproduction number, $R_0$, is the expected number of secondary infections in a susceptible population that arise directly from the primary infection of a single individual. For the SIR model described in Section \ref{sect.examples}, $R_0=\beta/\gamma$. For measles, $R_0$ is often reported as between $14$ and $18$ \cite[e.g.][]{Measles1982}, which fits with the World Health Organisation (WHO) recommendation of vaccination level of at least $93-95\%$ \cite[]{WHOMeasles}.

We fit the SIR model to the notification data for the Swansea LA provided in Table \ref{table.Swansea} so as to estimate the overall $R_0$ for the partially vaccinated population in Swansea and to demonstrate inference on the unknown number of infectious individuals at each observation time. In fitting the model we are making several assumptions and simplifications, including the following. Firstly, we are ignoring infections from Swansea to other LAs and from these LAs to Swansea; since most of the infections occurred in Swansea the former will outnumber the latter and so we will underestimate the `true' $R_0$, and provide a `local' $R_0$ at the epicentre of the infection. Secondly it is known that the lowest level of vaccination, and the highest level of infection was amongst 10-18 year olds \cite[]{Measles2013B}; a more accurate model would, therefore, partition the population into age groups. Age-stratified, continuous-time Markov chain SIR models are difficult to fit in general, however, and often a deterministic version of the model is used \cite[e.g.][]{BroadfootKeelingMRes}. Finally, we treat a notification as equivalent to a removal: this is not unreasonable as once an individual has been diagnised by a GP with suspected measles they will be asked to isolate themselves.

\begin{table}
    \begin{center}
  \caption{\small{Number of measles notifications in the Swansea Local Authority area by month (from \url{http://www.wales.nhs.uk/sitesplus/888/page/66389}, February 10th 2020). \label{table.Swansea}}}
  
  \begin{tabular}{c|ccc|ccccccc}
    &\multicolumn{3}{c|}{2012}&\multicolumn{7}{c}{2013}\\
Month  &Oct &Nov&Dec & Jan &Feb& Mar&Apr&May&Jun&Jul\\
\hline
  Day number&0&30&61&92&120&151&181&212&242&273\\
  \hline
  Notifications&0&10&27&34&59&183&278&56&17&0
  \end{tabular}
    \end{center}
\end{table}

%\subsubsection{Bayesian Analysis}
As described in Section \ref{sec.statespace.red} we add as latent variables the number of infections at each of the reporting times, Days 30, 61, 92, ..., 212. The number of infections at times 242 and 273 must both be zero.

To understand the evolution near the peak of the epidemic and speed up inference still further, we add latent observation times during the peak of the infection, at Days 125, 130, 135, 140, 145, 155, 160, 165, 170 and 175. This leads to $10$ further latent observations of the number of infected individuals and (because of constraints) $10$ further latent observations of the number removed during each reduced time period, leading to a total of $27$ integer latent variables.

We use a $\mathsf{N}(\log 5,2/3)$ prior for $\log R_0=\log (\beta/\gamma)$, a $\mathsf{N}(\log (1/15),1)$ prior for $\log \gamma$ and, because it is very poorly identified, we set the prior for the effective population size to $p(N_{pop}=n)\propto \exp(-n/500)1_{\{n\ge 1000\}}$.

We perform inference via a Metropolis-within-Gibbs algorithm: $\theta=(\log \beta, \log \gamma)$ is updated via a random walk proposal with a jump of $\mathsf{N}(0,\lambda^2 \mathsf{I}_2)$,  $n_{pop}$ via an integer-valued random walk proposal, and $x_{latent}$, the latent observations via integer random walks, with physical constraints (such as the sum of all the $R$s not being able to exceed $n_{pop}$) checked for automatically; see Appendix \ref{sec.SIRextra} for more details.

The basic reproduction number, $R_0$, is estimated as $1.15$, with a $95\%$ credible interval of $(1.01,1.31)$. This fits with other information known: firstly,up until 2013, $R_0$ only changed gradually over time (due to year-on-year variations in infant vaccination rates) and it cannot have reached much higher than $1$ in late 2012 as otherwise there would have been an outbreak in a previous year; secondly an $R_0$ of $1.15$ if the true $R_0$ is $16$, corresponds to a vaccination level of $93\%$, and $R_0=1.3$ corresponds to a $92\%$ level, and as argued earlier, we expect to slightly underestimate $R_0$. As of December 2012, the estimated coverage of one dose of MMR vaccine among 16 year-olds in Wales was $91\%$ \cite[]{COVER105}.

The right-hand panels of Figure \ref{fig.Moran} show the posterior median and $95\%$ credible intervals for the number of infections at each of the monthly observation times and at the 10 additional latent times, and similar intervals for the cumulative number of infections. In any infectious disease, at any current time point, it is vital to understand the current, unknown, number of infections in order to be able to predict the future course of an epidemic.

\section{Discussion}
\label{sec.discuss}
We have shown that inference, prediction and filtering for continuous-time Markov chains with a large but finite statespace, especially those arising from reaction networks is not just conceptually straightforward when the matrix exponential is used, but it is also often practical. We have provided and demonstrated the use of robust tools for this purpose in \texttt{R}, which opens up the direct use of and inference for reaction-network models to a wider audience. Straightforward inference for epidemic models, such as the SIR and SEIR models is particularly apposite at the time of submission, as it might have enabled an analysis of early COVID-19 infection data by people not expert in the more complex MCMC methodology typically used.

We emphasise that we are not suggesting that the tools we provide should replace the particle MCMC, ABC and SMC methods currently employed. In our experience, inference for epidemic models coded in a fast, compiled language is often more efficient in terms of effective samples per second, for example, than the approach using matrix exponentiation. However, the matrix exponential approach \emph{is} much more straightforward, and the code that uses it can be written in the simpler, interpreted language \texttt{R}.

As the size of the statespace increases, the efficiency of the matrix exponentiation approach decreases; however, once the statespace becomes sufficiently large, the evolution of the process is often approximated by a stochastic differential equation \cite[e.g.][]{Golightly05,fearnhead12} or, when the behaviour is effectively deterministic, by ordinary differential equations \cite[e.g.][]{BroadfootKeelingMRes}.

For the scaling and squaring approach, in particular, the cost of the exponentiation of $\matQ/K$ can be nearly halved by using a Pad\'e approximant \cite[e.g.][]{Dubious}, but this then requires a matrix inversion, and so, for reasons of robustness, was not pursued here.

\section*{Acknowledgements}
I would like to thank Prof. Lam Ho for suggesting that the reformulation of the statespace in \cite{HCS2018} in terms of births might be applicable within the methodology presented herein. I am also grateful to Dr. Andrew Golightly for several  useful discussions. 

\bibliographystyle{apalike}
\bibliography{matexp}

\appendix
\section{Evaluating $m_{\epsilon}(\rho)$}
\label{sec.mepsrho}
Our fast, robust and accurate method for evaluating $m_{\epsilon}(\rho)$, as defined in Section \ref{sec.uniformisation} relies on the following new result.
\begin{theorem}
  \label{thrm.asymp}
  If $\rho\le \epsilon$, $m_{\epsilon}(\rho)=0$, and if $\rho\le \epsilon^{1/2}$, $0\le m_{\epsilon}(\rho)\le 1$. More generally: $m_{\epsilon}(\rho)\le \lceil m_+\rceil$, where 
\begin{align}
  \label{eqn.mboundhi}
    m_+:= \rho-\frac{1}{3}\log \epsilon \left\{1+\left(1-\frac{18\rho}{\log \epsilon}\right)^{1/2}\right\}-1.
  \end{align}
Furthermore, 
\[
\lfloor m_-\rfloor
\le
m_{\epsilon}(\rho)
\le
\lceil m_{++}\rceil,
\]
where both inequalities require $\epsilon<0.04$ and the latter also requires $\epsilon<1-e^{-\rho}$ and $B>\log \epsilon$, where  
\begin{align}
  \label{eqn.mboundlo}
  m_-&:=\rho+\{2\rho\}^{1/2}\left\{-\log (\epsilon\sqrt{2\pi})-\frac{3}{2}\log A+\log(A-1)\right\}^{1/2},\\
  \label{eqn.mboundhihi}
  m_{++}&:=\rho+\frac{B-\log \epsilon}{3}\left\{1+\left(1+\frac{18\rho}{B-\log \epsilon}\right)^{1/2}\right\},
\end{align}
\[
A:=2\rho \mathsf{h}\left(\frac{m_+ +1}{\rho}\right)
~~~
\mbox{and}
~~~
B:= -\frac{1}{2}\log 4\pi\rho \mathsf{h}\left(\frac{m_-}{\rho}\right),
\]
and $\mathsf{h}(x)=x-1+x\log x$.
\end{theorem}

\begin{figure}
\begin{center}
  \includegraphics[scale=0.40,angle=0]{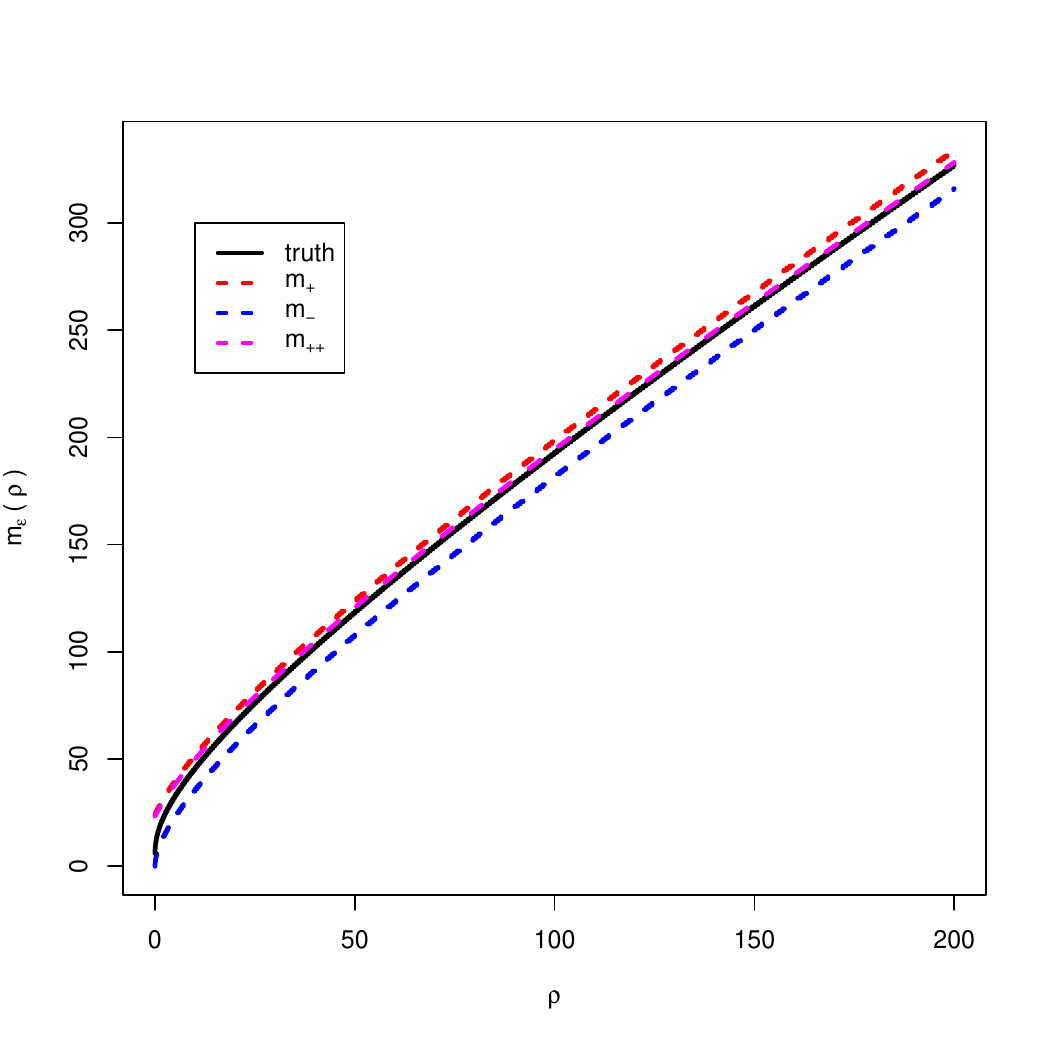}
  \includegraphics[scale=0.40,angle=0]{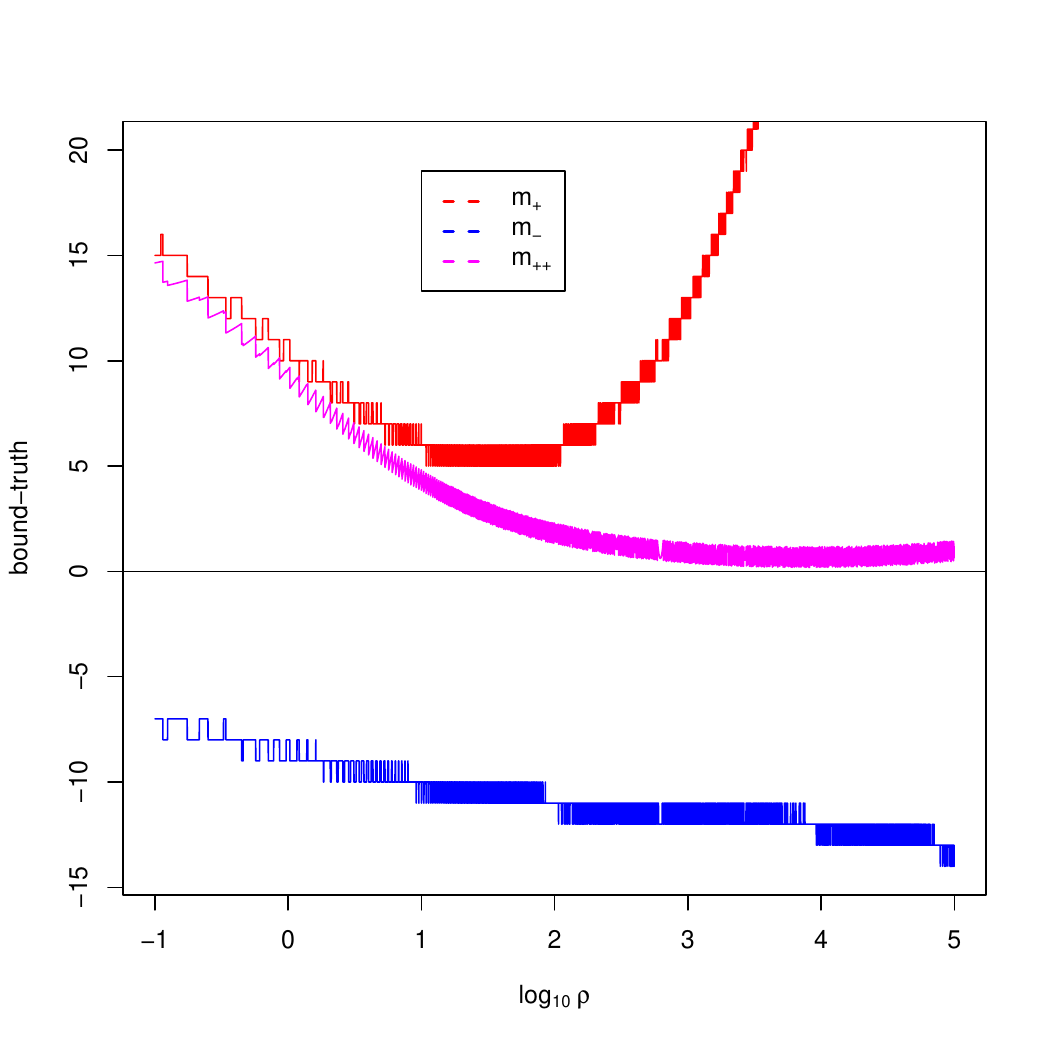}
\caption{Left panel: $m_{\epsilon}(\rho)$ together with its upper and lower bounds from Theorem \ref{thrm.asymp}, plotted against $\rho$ for $\epsilon=10^{-16}$.
 Right panel $\mathsf{bound}(\rho)-m_{\epsilon}(\rho)$ against $\log_{10} \rho$ for $\epsilon=10^{-16}$. \label{fig.m_eps_rho}
}
\end{center}
\end{figure}

The bound \eqref{eqn.mboundhi} arises from a standard argument, whereas those in \eqref{eqn.mboundlo} and \eqref{eqn.mboundhihi} are derived from extremely sharp but intractable bounds on $r_m(\rho):=\Prob{\mathsf{Poisson}(\rho)>m}$ in \cite{Short13}; our bounds use only elementary functions and so are much quicker to compute than the quantile upper bound in \cite{Short13}, yet from Figure \ref{fig.m_eps_rho} they are still very sharp. 
The bounds in \eqref{eqn.mboundlo} and \eqref{eqn.mboundhihi} together with the alternative form in \eqref{eqn.alt.form}
\begin{equation}
\label{eqn.alt.form}
  r_m(\rho)=\frac{1}{\Gamma(m+1)}\int_0^\rho x^{m}e^{-x}\md x,
\end{equation}
which follows from the equivalence between at least $m+1$ events of a Poisson process with a unit rate occuring by time $\rho$ and the time until the $m+1$th event being at most $\rho$, permit a simple but fast binary search for $m_{\epsilon}(\rho)$.

 \subsection{Implementation details}
\label{sect.implement.merho}
Our binary search algorithm homes in on the required $m$ using the upper and lower bounds of Theorem \ref{thrm.asymp} together with the identity
\eqref{eqn.alt.form}, the right hand side of which can be evaluated quickly and accurately using the standard \texttt{C++} toolbox, \texttt{boost}. 
This is quicker than the standard implementation of the Poisson quantile function (e.g. as implemented in \texttt{boost}), which uses the Cornish-Fisher expansion  to approximate the quantile (hence needing an expensive evaluation of $\Phi^{-1}$) and then conducts a local search. 

\subsection{Proof of Theorem \ref{thrm.asymp}}
\label{proof.asymp.bound}
The simple bounds for small $\rho$ arise because $e^{-\rho}>1-\rho$. Hence $r_0(\rho)=1-e^{-\rho}<\rho$ and if $\rho\le\epsilon$, $m_{\epsilon}(\rho)=0$. Furthermore,
$r_1(\rho)=1-e^{-\rho}(1+\rho)<\rho^2$, so if $\rho \le \sqrt{\epsilon}$ then $r_1(\rho)\le \epsilon$, so $m_{\epsilon}(\rho)\le 1$.

The other bounds all use aspects of the following result.
\begin{lemma}
  \label{lemma.h}
  Let $\sfh(x):=1-x+x\log x$, then for $x\ge 1$,
  \[
  \frac{3}{6+2(x-1)}(x-1)^2\le \sfh(x)\le \frac{1}{2}(x-1)^2.
  \]
\end{lemma}
\begin{proof}
  The left hand inequality holds for $x\ge 0$ and is from \cite{Boucheron2013} page 36.
  For the right hand inequality, set $g(x)=(x-1)^2/2$ and notice that
  $0=\sfh(1)=g(1)=\sfh'(1)=g'(1)$, and $\sfh''(x)=1/x\le 1=g''(x)$ for $x\ge 1$.
\end{proof}

The first upper bound on $m_{\epsilon}(\rho)$, \eqref{eqn.mboundhi}, arises from a standard Chernoff argument \cite[e.g.][]{Boucheron2013} to the right tail of a $\mathsf{Poisson}(\rho)$ random variable, $X$. The moment generating function is $M_X(t)=\Expect{e^{Xt}}=\exp[\rho (e^t-1)]$, and by Markov's inequality:
\[
\Prob{X\ge m}=\Prob{e^{Xt}\ge e^{mt}}
\le
e^{-mt}M_{X}(t)
=e^{-mt+\rho (e^t-1)}.
\]
The inequality holds for all $t$ and the right-hand side is minimised at $t=\log(m/\rho)$, giving
\[
\Prob{X\ge m}\le \exp[-\rho \sfh(m/\rho)]\le \exp\left[-\rho\frac{3(m/\rho-1)^2}{6+2(m/\rho-1)}\right]
\]
by Lemma \ref{lemma.h}. Setting $\epsilon=\Prob{X\ge m+1}$ and $y=(m+1)/\rho-1$ gives $3\rho y^2 (6+2y)\log \epsilon\ge 0$, from which
$y\ge -\log \epsilon\times\sqrt{1-18\rho/\log \epsilon}/(3\rho)$,
and \eqref{eqn.mboundhi} follows on substituting for $y$.

The much tighter bounds in \eqref{eqn.mboundlo} and \eqref{eqn.mboundhihi} use
 Theorem 2 of \cite{Short13}, which can be rewritten to state that
\begin{align}
\label{eqn.Short}
  \Phi\left(-\sqrt{2\rho \sfh(m'/\rho)}\right)< \Prob{X> m}<
\Phi\left(-\sqrt{2\rho \sfh(m/\rho)}\right),
\end{align}
where $m':=m+1$ and where the left hand side holds provided $m'>\rho$ and the right hand side holds provided $m>\rho$. We first show that these conditions are satisfied. Firstly, when $\rho<1$, clearly $m'>\rho$, moreover $r_0(\rho)=1-e^{-\rho}$, so provided $1-e^{-\rho}>\epsilon$, we require $m\ge 1>\rho$. When $\rho\ge 1$, we use the easily verified facts that $r_{m}(m)$ is an increasing function of $m$ and $r_{m}(\rho)$ is an increasing function of $\rho$; thus for $\rho \ge m \ge 1$, $r_m(\rho)\ge r_{m}(m)\ge r_1(1)=1-2e^{-1}>0.04$, and the tolerance condition is not satisfied. We, therefore need $m>\rho$ (which also gives $m'>\rho$).

Neither $\Phi^{-1}$ nor $\sfh^{-1}$ is tractable (functions that perform $\Phi^{-1}(p)$ solve $\Phi(x)=p$ iteratively), and even with the bounds on $\sfh$ from Lemma \ref{lemma.h} and standard bounds on $\Phi$ in terms of $\phi$, tractable inversion is still not possible. We use the bound \eqref{eqn.mboundhi} to create \eqref{eqn.mboundlo}, and then \eqref{eqn.mboundlo} to create \eqref{eqn.mboundhihi}.

To prove \eqref{eqn.mboundlo}, since $\epsilon\le 0.04$, from the left inequality in \eqref{eqn.Short},
\[
0.04 \ge \Prob{X\ge m}\Rightarrow \sqrt{2\rho \sfh(m'/\rho)}\ge
-\Phi^{-1}(\epsilon)\approx 1.75> \sqrt{3}.
\]
Firstly, since $m_++1\ge m+1$, this ensures $A>1$, so $\log(A-1)$ is real. More importantly, it ensures that
 $[2\rho \sfh(m'/\rho)]^{-1/2}-[2\rho \sfh(m'/\rho)]^{-3/2}$ is a decreasing function of $[2\rho \sfh(m'/\rho)]^{1/2}$ and, since $\sfh'(x)>0$ for $x>1$, it is also a decreasing function of $m'$.
The $m'$ that we desire satisfies $m'\le m_++1=:m_+'$, and hence 
\[
  [2\rho h(m'/\rho)]^{-1/2}-[2\rho h(m'/\rho)]^{-3/2}
  \ge
  [2\rho h(m_+'/\rho)]^{-1/2}-[2\rho h(m_+'/\rho)]^{-3/2}.
\]
Since, for $y>0$, $\Phi(-y)>(1/y-1/y^3) \phi(y)$,
\[
\Phi\left(-\sqrt{2\rho h(m'/\rho)}\right)
\ge
   \left\{[2\rho h(m_+'/\rho)]^{-1/2}-[2\rho h(m_+'/\rho)]^{-3/2}\right\}
   \phi\left(\sqrt{2\rho h(m'/\rho)}\right).
%   =
%   \frac{1}{\sqrt{2\pi}}
%   [2\rho h(m_+'/\rho)]^{-1/2}-[2\rho h(m_+'/\rho)]^{-3/2}
%\exp[-\rho h(m'/\rho)].   
\]
Combining the left inequality in \eqref{eqn.Short} with the right-hand inequality in Lemma \ref{lemma.h} gives
\[
\epsilon \ge
   \frac{1}{\sqrt{2\pi}}
   \left\{[2\rho h(m_+'/\rho)]^{-1/2}-[2\rho h(m_+'/\rho)]^{-3/2}\right\}
\exp\left[-\frac{(m'-\rho)^2}{2\rho}\right].   
\]
Equation \eqref{eqn.mboundlo} follows on rearrangement.

To show \eqref{eqn.mboundhihi} we apply the right hand inequality in \eqref{eqn.Short} and the bound $\Phi(-x)< \phi(x)/x$, then the fact that
$m\ge m_-$, and finally Lemma \ref{lemma.h} to find:
\begin{align*}
  \Prob{X> m}&<
  \frac{1}{\{4\pi\rho \sfh(m/\rho)\}^{1/2}}
\exp\left[-\rho\sfh(m/\rho)\right]
  \le
  \frac{1}{\{4\pi\rho \sfh(m_-/\rho)\}^{1/2}}
  \exp\left[-\rho\sfh(m/\rho)\right] \\
  &\le
  \frac{1}{\{4\pi\rho \sfh(m_-/\rho)\}^{1/2}}
  \exp\left[-3\rho \frac{(x-1)^2}{6+2(x-1)}\right],  
\end{align*}
where $x=m/\rho$. We must, therefore, ensure that the final bound is no more than $\epsilon$. Rearranging this gives 
$3\rho(x-1)^2-2(B-\log \epsilon)(x-1)-6(B-\log \epsilon)\le 0$, so that when $B-\log \epsilon>0$, $x-1\le (B-\log \epsilon)(1+\sqrt{1+18\rho/(B-\epsilon)})/(3\rho)$.

\section{Functions in the \texttt{expQ} package}
\label{sec.expQ}
The functions in the \texttt{expQ} package are provided below. Each function requires a rate matrix, $Q$, which can be sparse or dense, and a precision, $\epsilon$.\\
\texttt{Unif\_v\_exp\_Q} takes a horizontal vector, $v$, and calculates $v\exp \matQ$ via uniformisation.\\
\texttt{SS\_v\_exp\_Q} takes a horizontal vector, $v$, and calculates $v\exp \matQ$ via scaling and squaring.\\
\texttt{v\_exp\_Q} takes a horizontal vector, $v$, and calculates $v\exp \matQ$ via whichever is likely (based on empirical results on an i7-3770 CPU) to be the more efficient of uniformisation or scaling and squaring.\\
\texttt{vT\_exp\_Q} takes a vertical vector, $v$, and calculates $\left(v^\top\exp \matQ\right)^\top$ via whichever is likely (based on empirical results on an i7-3770 CPU) to be the more efficient of uniformisation or scaling and squaring.\\
\texttt{SS\_exp\_Q} calculates $\exp \matQ$ using scaling and squaring.

\section{Latent-variable updates for the SIR model}
\label{sec.SIRextra}
Our particular reduced-statespace implementation of the SIR model fit for the Swansea Measles epidemic uses $10$ additional latent observation times, $5$ between days $120$ and $151$ (at days $125,130,135,140$ and $145$) and five between days $151$ and $212$ (at days $155,160,165,170$ and $175$). This leads to $27$ latent variables: $17$ unknown number of infecteds at the (true and latent) observation times and $10$ (not $12$ because two sums are known) unknown numbers of recovered for the time period since the previous (true or latent) observation time. We emphasise that the $R$ latent variables are \emph{not} the cumulative number of recovered individuals since the epidemic began.

When a new latent vector is proposed, we first check whether it can possibly fit with the current $n_{pop}$ and the known data. If it does not fit, then the proposal may be rejected without any matrix exponentiation. At the $j$th (true or latent) time point, denote the current number of infecteds by $I_j$ and the number removed since the previous time point by $R_j$. Let $J$ be the total number of (true and latent) time points. Note that $S_j=n_{pop}-I_0-I_j-\sum_{i=1}^j R_i$. The following checks are performed:
\begin{enumerate}
\item For each $j=1,\dots,J$: $I_j\ge 0$, $R_j \ge 0$ and $S_j\ge 0$.
\item For each $j=1,\dots,J$: $S_{j}\le S_{j-1}$.
\item For each $j=1,\dots,J-1$: $I_j\le \sum_{i=j+1}^JR_j$.
\end{enumerate}

These constraints can hinder the mixing of the integer-valued random walk algorithm on the latent variables, so we split the latent variables into four groups, grouped by observation time. This grouping has the additional advantage that only a subset of matrix exponentiation calculations need be performed for each of the four individual proposals.

\section{Log-likelihood \texttt{R} code for the Moran model}
\label{app.code}
To indicate the simplicity of inference via the matrix exponential, we provide code to evaluate the log-likelihood for the Moran model. Code for the filtering distribution is very similar but there is no need to track the re-normalisation constant (in $\texttt{ll}$).

\begin{verbatim}
## Log likelihood for Moran model
## thetaunk=(log alpha, log beta, logit u, logit v)
## npop=known population size
## obstim=vector of observation times
## yobs=vector of observations
## errn=parameter for Binom(2*errn,0.5)-errn error distribution
## nu=t(rep(1/d,d)); ## uniform prior over statespace 
getll<-function(thetaunk,npop,obstim,yobs,errn,nu) {
    thetas=c(exp(thetaunk[1:2]),exp(thetaunk[3:4])/(1+exp(thetaunk[3:4])),npop)
    nobs=length(obstim)
    d=npop+1 ## size of statespace; states are 0, ..., npop

    Q=MoranGetQ(thetas) ## same Q every time as whole statespace
    ll=0
    nu=nu*dbinom(yobs[1]+errn-(0:npop),2*errn,0.5)
    for (i in 2:(nobs-1)) {
        currtot=sum(nu)
        if ((currtot<1e-6)||(currtot>1e6)) { ## avert possible over/underflow
            nu=nu/currtot
            ll=ll+log(currtot)
        }
        nu=Unif_v_exp_Q(nu,Q*(obstim[i+1]-obstim[i]),1e-15)
        nu=t(as.vector(nu)*dbinom(yobs[i+1]+errn-(0:npop),2*errn,0.5))
    }
    ll=ll+log(sum(nu))

    return(ll) 
}
\end{verbatim}
\end{document}